\documentclass{article}
\usepackage{amsmath,amsfonts,amsthm,color}
\usepackage{color}
\usepackage{graphicx}
\usepackage{hyperref}
\newtheorem{assumption}{Assumption}[section]
\newtheorem{corollary}{Corollary}[section]
\newtheorem{proposition}{Proposition}[section]

\newtheorem{theorem}{Theorem}[section]
\theoremstyle{remark}
\newtheorem{remark}{Remark}[section]
\usepackage[normalem]{ulem}

\usepackage{geometry}
\usepackage{enumerate}

\usepackage{float}

\newcommand{\ind}{1\hspace{-2.1mm}{1}} %Indicator Function

\DeclareMathOperator*{\argmin}{\arg\!\min}

\author{Aitor Muguruza\thanks{The author is grateful to Luc Mathieu and Sylvain Pouteaux for fruitful discussions} \\
 Department of Mathematics, Imperial College London and NATIXIS
\\
aitor.mgz@gmail.com
}
\date{\today}
\title{Not so Particular about Calibration:\\ Smile Problem Resolved}

\begin{document}
\maketitle
\begin{abstract}
We present a novel Monte Carlo based LSV calibration algorithm that applies to all stochastic volatility models, including the non-Markovian rough volatility family. Our framework overcomes the limitations of the particle method proposed by Guyon and Henry-Labord\`ere (2012) and theoretically guarantees a variance reduction without additional computational complexity. Specifically, we obtain a closed-form and exact calibration method that allows us to remove the dependency on both the kernel function and bandwidth parameter. This makes the algorithm more robust and less prone to errors or instabilities in a production environment. We test the efficiency of our algorithm on various hybrid (rough) local stochastic volatility models.
\end{abstract}
\section{Introduction}
The calibration of local-stochastic volatility (LSV) models is a fundamental problem in financial modelling.  In the case of low dimensional Markovian models, efficient PDE methods are already available to compute the leverage function (see Lipton \cite{Lipton}). For more complex Markovian models, (e.g. multi factor or hybrid models)  PDE methods fall into the so-called curse of dimensionality, making their use more complex and slow. The most transparent and effective solution to this LSV calibration problem in high dimensional models has been the particle method developed by Guyon and Henry-Labord\`ere \cite{Guyon-Labordere}. It is worth noting, however, that this method requires the prespecification of a kernel function and bandwidth parameter, which can be unnnerving in an automatised production environment.\\\\
In this paper we present an exact algorithm that follows the principle of the particle method, without relying on non-parametric methods to obtain the leverage function. Perhaps surprisingly, we are able to obtain an exact and closed-form algorithm that calibrates the leverage function, and applies to all Stochastic Volatility (SV) models.\\\\
The article is organised as follows. We first introduce our LSV framework and present the principle of LSV calibration algorithms. Then, we present our theoretical results which yield the new exact algorithm. Finally, we illustrate the efficiency of our algorithm on various models; notably, we achieve the calibration of a first of a kind hybrid rough LSV model. 
\section{LSV framework and mathematical setting}
A LSV model under the risk-neutral measure  $\mathbb{Q}$ is given by
\begin{equation}\frac{\mathrm{d}S_t}{S_t}=r_t \mathrm{d}t+\lambda(S_t,t) \sqrt{V_t} \mathrm{d}W_t\label{eq:LSV}
\end{equation}
where $\lambda:\mathbb{R}\times \mathbb{R}_+\to \mathbb{R}_+$ is the so-called leverage function. $W$ here, is a unidimensional Brownian motion and without loss of generality, the processes $V$ and $r$ are assumed to be driven by an n-dimensional Brownian motion $Z$, where $n\in\mathbb{N}$. In addition, we
assume the tuple $(W_t,Z_t)$  to have the following correlation structure: 

$$\Sigma=\begin{pmatrix} 1 & \rho^T_{WZ}\\ \rho_{WZ} & \Sigma_Z\end{pmatrix}\in\mathbb{R}^{(n+1)\times(n+1)},\quad \Sigma_Z\in\mathbb{R}^{n\times n}, \quad  \rho_{WZ}\in\mathbb{R}^{n\times 1}.$$
The processes are defined on a given probability space $\left (\Omega,(\mathcal{G}_t)_{t\geq 0},\mathbb{Q}\right)$ with $\mathcal{G}_t$ being the natural filtration
generated by the aforementioned Brownians. Let us further denote by $\mathcal{F}^W_t$ and  $\mathcal{F}^Z_t$ the filtrations generated by $W$ and $Z$ respectively, such that $\mathcal{G}_t=\mathcal{F}^W_t\cup\mathcal{F}^Z_t$ holds. The existence of solutions of this McKean SDE is a very intricate mathematical problem (see \cite{Jourdain}) and falls
outside the scope of this paper.
\begin{remark}
The precise definition of filtrations and correlations goes beyond sheer mathematical rigour, as these are the essential tools to develop the exact formula of our algorithm later on.
\end{remark}
\begin{remark}
Dividends are neglected for sake of simplicity, but it is straightforward to consider them as in Guyon and Henry-Labord\`ere \cite{Guyon-Labordere}.
\end{remark}
\section{The calibration of the leverage function in a nutshell}
The probabilistic condition in \eqref{eq:LSV} for the leverage function to be calibrated to market smiles is given
by (see Balland \cite{Balland})
\begin{equation}\lambda^2(K,t)\frac{\mathbb{E}^{\mathbb{Q}}[D(t)V_t|S_t=K]}{\mathbb{E}^{\mathbb{Q}}[V_t|S_t=K]}=\sigma_{Dup}(K,t)^2-\frac{\mathbb{E}^{\mathbb{Q}}\left[D(t)\left(r_t-\bar{r}_t\right)\ind_{\{S_t>K\}}\right]}{\frac{1}{2}K\frac{\partial^2 C}{d K^2}},\label{eq:LeverageFunctionCondition}\end{equation}
where $D(t)=e^{-\int_0^t r_s ds}$, $ZCB(t)=\mathbb{E}^{\mathbb{Q}}[D(t)]$, $\bar{r}_t=\frac{\partial \log(ZCB(t))}{dt}$. Note that $\mathbb{E}^{\mathbb{Q}}\left[D(t) r_t \right]=\bar{r}_t$. We define $\sigma_{Dup}(K,t)$ to be the local volatility (Dupire \cite{Dupire})
associated with the observed market i.e.

$$\sigma_{Dup}(K,T)^2=\frac{\frac{\partial C}{\partial T}+K\frac{\partial C}{dK}\bar{r}_0}{\frac{1}{2}K^2 \frac{\partial^2 C}{\partial K^2}}.$$
All numerical algorithms known to solve the calibration task  rely on the following standard assumption:
\begin{assumption}\label{ass:Lambda}
The leverage function has the following structure:

$$\lambda(x,t)=\sum_{i=1}^{\tau(t)} f_i(x) \ind_{\{t\in[t_{i-1},t_i)\}},\quad \tau(t):=\argmin_{j\in\mathbb{N}}\{t<t_j\}$$  for  functions $f_i:\mathbb{R}_+\mapsto \mathbb{R}_+.$
\end{assumption}
Under Assumption \ref{ass:Lambda}, the time domain is discretised and the following algorithm is constructed:
\begin{figure}
\centering
\includegraphics[scale=0.6]{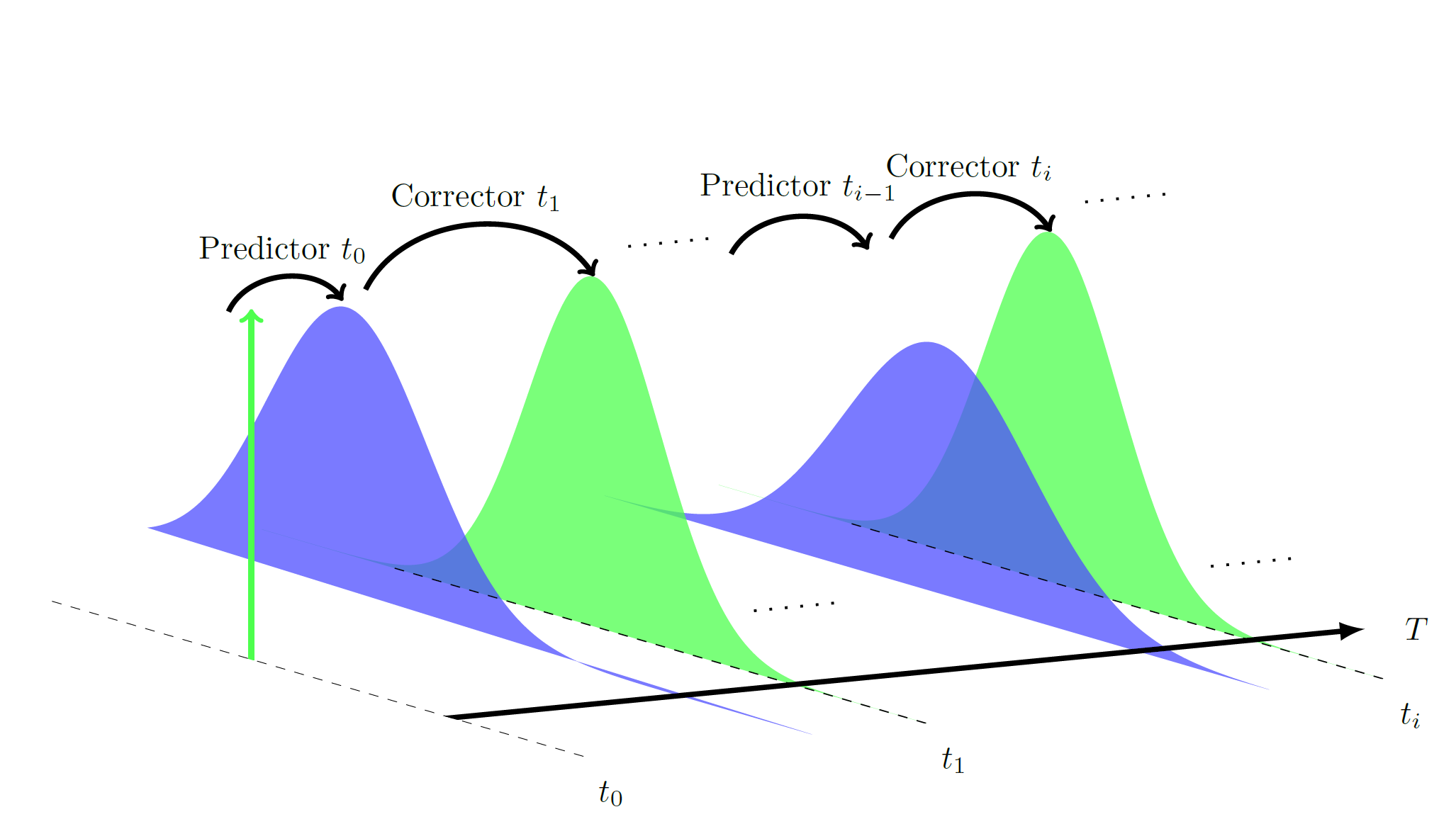}
\caption{Graphical description of the forward in time LSV calibration algorithm. The $t_0$ distribution of the Spot is by convention a Dirac function, hence the arrow representation. The algorithm predicts the distribution implied by the model at time $t_{i-1}$ for time $t_i$ and corrects accordingly to match the market implied distribution.}
\label{fig:PredictorCorrector}
\end{figure}\\\\
\noindent\textbf{LSV Calibration Algorithm}
\begin{enumerate} 
\item Fix a time partition $\{t_0=0,...,t_n=T\}$ and space partition $\{K_0=K_{min},...,K_m=K_{max}\}$.
\item Initialize $\lambda(K_j,t_0)=\displaystyle\frac{\sigma_{Dup}(K_j,t_0)}{V_0},\quad j=0,...,m$ .
\item For $i=1...n:$  \\
\hspace*{0.25cm} For $j=0,...,m$
\begin{quote}
\textbf{Step 1. Predictor:}
\begin{quote}
$\bullet$ Calculate $\mathbb{E}^{\mathbb{Q}}\left[ D(t_i)\left(r_{t_i}-\bar{r}_{t_i}\right)\ind_{\{S_{t_i}>K_j\}}\right]$\\
$\bullet$ Obtain $\displaystyle \frac{\mathbb{E}^{\mathbb{Q}}[D(t_i)V_{t_i}|S_{t_i}=K_j]}{\mathbb{E}^{\mathbb{Q}}[D(t_i)|S_{t_i}=K_j]}$
\end{quote}
\textbf{Step 2. Corrector:}
\begin{quote}
$\bullet$$\lambda^2(K_j,t_i)=\displaystyle\left(\sigma_{Dup}(K_j,t_i)^2-\frac{\mathbb{E}^{\mathbb{Q}}\left[ D(t_i) \left(r_{t_i}-\bar{r}_{t_i}\right)\displaystyle \ind_{\{S_{t_i}>K_j\}}\right]}{\frac{1}{2}K\frac{\partial^2 C}{d K^2}}\right)\frac{\mathbb{E}^{\mathbb{Q}}[ D(t_i)|S_{t_i}=K_j]}{\mathbb{E}^{\mathbb{Q}}[ D(t_i)V_{t_i}|S_{t_i}=K_j]}.$
\end{quote}
\end{quote}
\end{enumerate}

As noted above, the procedure updates the leverage function $\lambda$, the critical step being the correct computation of $ \frac{\mathbb{E}^{\mathbb{Q}}[D(t_i)V_{t_i}|S_{t_i}=K_j]}{\mathbb{E}^{\mathbb{Q}}[D(t_i)|S_{t_i}=K_j]}$ at time $t_{i-1}$. Figure \ref{fig:PredictorCorrector} is a graphical description of the iterative procedure (forward in time) of the calibration algorithm.

\subsection{The original particle method} 
In order to compute  conditional expectations, Guyon and Henry-Labord\`ere \cite{Guyon-Labordere} propose a Monte Carlo based estimation via non-parametric Nadaraya-Watson kernel regression, more precisely if we denote by $S^u_{t_i}$ and $v^u_{t_i}$ the u-th particle obtained in a Monte Carlo simulation with $M$ particles, then the estimator is given by:
\begin{equation}\label{eq:ParticleMethod}\frac{\mathbb{E}^{\mathbb{Q}}[ D(t_i)|S_{t_i}=K_j]}{\mathbb{E}^{\mathbb{Q}}[ D(t_i)V_{t_i}|S_{t_i}=K_j]}\approx \frac{\displaystyle\sum_{u=1}^{M} D(t_i)^u V^u_{t_i}K_h(S^u_{t_i}-K_j)}{\displaystyle\sum_{u=1}^{M}D(t_i)^u K_h(S^u_{t_i}-K_j)},\end{equation}
where $K_h(\cdot)$ is a suitable kernel function with bandwith parameter $h>0$. In the original article, a rule of thumb is given for the choice of $h$ and the choice of a quartic kernel is recommended. In our numerical tests in Section \ref{sec:numerics}, we follow these guidelines for implementing the particle method as benchmark.
\subsection{Limitations of the particle method: Bias, Variance and Convergence}
It is well known, that non-parametric kernel regression methods suffer from a bias of order $\mathcal{O}(h^2)$ (see H\"ardle and Bowman \cite{HardleBowman}). Most importantly, the variance of the estimator is of order $\mathcal{O}(n^{-1/2+p(h)})$ where $p(h) < 0$ heavily depends on the correct choice of the bandwidth $h$. The need for a prespecified parameter $h$ of such critical importance for the correct performance of the method, poses an inherent risk and makes it inclined to potential instabilities.

\section{ The exact formula}

In order to obtain the closed-form formula for $\frac{\mathbb{E}^{\mathbb{Q}}[ D(t_i)|S_{t_i}=K_j]}{\mathbb{E}^{\mathbb{Q}}[ D(t)V_{t}|S_{t}=K]}$ , we first obtain the following result.

\begin{theorem}\label{thm:condDistr}
Given model \eqref{eq:LSV} and Assumption \ref{ass:Lambda} we have
$$\log S_{t_i}|\mathcal{F}^{W}_{t_{i-1}}\cup \mathcal{F}^{Z}_{t_i}\sim \mathcal{N}\left(\mu_i, (1-\hat{\rho}^2)\sigma^2_i \right),$$
where \begin{align*}\mu_i&:=\log S_{t_{i-1}}+\int_{t_{i-1}}^{t_i} r_u  \mathrm{d}u
+\int_{t_{i-1}}^{t_i} \sqrt{V_u}\lambda(S_{t_{i-1}},u)   \mathrm{d}W_u^{||} -\frac{1}{2} \int_{t_{i-1}}^{t_i}V_u \lambda(S_{t_{i-1}},u)^2 \mathrm{d}u\\
\sigma^2_i&:= \int_{t_{i-1}}^{t_i}V_u \lambda(S_{t_{i-1}},u)   \mathrm{d}u,\quad 
\hat{\rho}^2:=\rho_{WZ}^T \Sigma_Z^{-1} \rho_{WZ},\quad 
W_t^{||}:=\rho_{WZ}^T \Sigma_Z^{-1}Z_t 
\end{align*}
\end{theorem}
\begin{proof}
First we see that,
$$\log S_{t_i}|\mathcal{F}^{W}_{t_{i-1}}\cup \mathcal{F}^{Z}_{t_i}\sim \mu_i +\sqrt{1-\hat{\rho}^2}\int_{t_{i-1}}^{t_i} \sqrt{V_u}\lambda(S_{t_{i-1}},u)  \mathrm{d}X_u$$
where $X$ is a Brownian motion independent of $\mathcal{F}^{W}_{t_{i-1}}\cup \mathcal{F}^{Z}_{t_i}$. We remark $\lambda(S_t,t)$ is constant and $\mathcal{G}_{t_i}$-measurable for $t\in[t_{i-1},t_i)$. Note also that $\mu_i\in\mathcal{F}^{W}_{t_{i-1}}\cup \mathcal{F}^{Z}_{t_i}$ and $(1-\hat{\rho}^2)\int_{t_{i-1}}^{t_i} \sqrt{V_u}\lambda(S_{t_{i-1}},u)  \mathrm{d}X_u$ has a deterministic integrand under conditioning, which by properties of the It\^o integral is a centred Gaussian with variance $(1-\hat{\rho}^2)\sigma_i^2$ and the result follows.
\end{proof}
The next corollary gives the exact formula that we proclaimed.
\begin{corollary}\label{corollary}
Given model \eqref{eq:LSV} and Assumption \ref{ass:Lambda} we have
\begin{equation}\frac{\mathbb{E}^{\mathbb{Q}}[ V_{t_i} D(t_i)|S_{t_i}=K_j]}{\mathbb{E}^{\mathbb{Q}}[ D(t_i)|S_{t_i}=K_j]}=\frac{\mathbb{E}^{\mathbb{Q}}\left[D(t_i) V_{t_i}\displaystyle\frac{e^{-\frac{1}{2}(d_i(K))^2}}{\sqrt{\sigma_i^2}}\right]}{\mathbb{E}^{\mathbb{Q}}\left[D(t_i) \displaystyle\frac{e^{-\frac{1}{2}(d_i(K))^2}}{\sqrt{\sigma_i^2}}\right]},\quad d_i(K)=\frac{\mu_i-\log(K)}{\sqrt{(1-\hat{\rho}^2)\sigma^2_i}}. \label{eq:condExact}\end{equation}
\end{corollary}
\begin{proof}
We note that $$\mathbb{E}^{\mathbb{Q}}[D(t_i) V_{t_i}|S_{t_i}=K]=\frac{\mathbb{E}^{\mathbb{Q}}[D(t_i) V_{t_i}\delta(S_{t_i}-K)]}{\mathbb{E}^{\mathbb{Q}}[\delta(S_{t_i}-K)]},$$
where $\delta(\cdot)$ represents the Dirac function at $0$. Using the Tower property we further get 
$$\mathbb{E}^{\mathbb{Q}}[D(t_i) V_{t_i}|S_{t_i}=K]=\frac{\mathbb{E}^{\mathbb{Q}}\left[D(t_i) V_{t_i}\mathbb{E}^{\mathbb{Q}}[\delta(S_{t_i}-K)|\mathcal{F}^{W}_{t_{i-1}}\cup \mathcal{F}^{Z}_{t_i}]\right]}{\mathbb{E}^{\mathbb{Q}}\left[\mathbb{E}^{\mathbb{Q}}[\delta(S_{t_i}-K)|\mathcal{F}^{W}_{t_{i-1}}\cup \mathcal{F}^{Z}_{t_i}]\right]},$$
where we have used that $V_{t_i},D(t_i)\in\mathcal{F}^{Z}_{t_i}$. Next, we invoke Theorem \ref{thm:condDistr}, which gives the conditional probability density function $\phi(\cdot)$ of $S$:
$$\phi(x)=\displaystyle\frac{1}{x}\frac{e^{-\frac{1}{2}(d_i(x))^2}}{\sqrt{2\pi \sigma_i}}.$$
The result can be now directly derived by using the conditional density along with the Dirac function and the same procedure for $\mathbb{E}^{\mathbb{Q}}[D(t_i)|S_{t_i}=K]$.
\end{proof}
In a Monte Carlo setting we may apply Corollary \ref{corollary} to obtain
\begin{equation}\label{eq:MCexact}\frac{\mathbb{E}^{\mathbb{Q}}[D(t_i) V_{t_i}|S_{t_i}=K]}{\mathbb{E}^{\mathbb{Q}}[D(t_i) |S_{t_i}=K]}=\frac{\displaystyle \sum_{u=1}^M D(t_i)^u V^u_{t_i}\frac{e^{-\frac{1}{2}(d^u_i(K))^2}}{\sqrt{(\sigma^u_i)^2}}}{\displaystyle \sum_{u=1}^MD(t_i)^u\frac{e^{-\frac{1}{2}(d^u_i(K))^2}}{\sqrt{(\sigma^u_i)^2}} }\end{equation}
where $u$ represents the $u$-th particle.
\begin{remark}
Note that expression \eqref{eq:MCexact} allows for exact (unbiased) Monte Carlo computation of
the conditional expectation, without the use of external parameters. Remarkably, the use of conditional expectations in \eqref{eq:MCexact} theoretically guarantees a variance reduction.
\end{remark}
\subsection{Further exploiting the discrete nature of numerical algorithms: lognormal SV case}
The results developed in Corollary \ref{corollary} allow to obtain a closed-form expression for virtually any SV model. In spite of the universality of the result, there is a (mild) memory cost of storing $\int_{t_{i-1}}^{t_i} \sqrt{V_u}\lambda(S_{t_{i-1}},u)  dW_u^{||}$. In this section, we show how to remove this dependence. For simplicity, in this section we assume that the instantaneous variance in \eqref{eq:LSV} follows a lognormal distribution
\begin{equation}\label{eq:Vlognormal}
\log(V_{t_i})\sim \mathcal{N}(\xi_i,\nu_i).
\end{equation}
Hence, the conditional distribution is also given by a lognormal random variable:
\begin{equation}
\log(V_{t_i})|\mathcal{G}_{t_{i-1}}\sim \mathcal{N}(\widetilde{\xi}_i,\widetilde{\nu}_i).
\end{equation}
Furthermore, we consider the processes $V$ and $r$ to be piecewise constant, i.e.

\begin{equation}\label{eq:discreteTime}V_t:=\sum_{i=1}^{\tau(t)} V_{t_{i-1}} \ind_{\{t\in[t_{i-1},t_i)\}},\quad r_t:=\sum_{i=1}^{\tau(t)} r_{t_{i-1}} \ind_{\{t\in[t_{i-1},t_i)\}},\quad  \tau(t):=\argmin_{j\in\mathbb{N}}\{t<t_j\}. \end{equation}
Note that \eqref{eq:discreteTime} is the usual time discretisation needed to perform numerical (forward Euler) simulation or
PDE pricing; it therefore does not pose additional constraints.
\begin{proposition}\label{proposition}
Under Assumption \ref{ass:Lambda} and model \eqref{eq:LSV} with $V$ and $r$ given by \eqref{eq:Vlognormal}-\eqref{eq:discreteTime} we have,
\begin{equation}\label{eq:CloseFormLognormal}\frac{\mathbb{E}^{\mathbb{Q}}[ V_{t_i} D(t_i)|S_{t_i}=K_j]}{\mathbb{E}^{\mathbb{Q}}[ D(t_i)|S_{t_i}=K_j]}=\frac{\mathbb{E}^{\mathbb{Q}}\left[ D(t_i) \frac{\exp\left\lbrace \widetilde{\xi}_i+\frac{1}{2}\widetilde{\nu}_i-\frac{1}{2\sigma_i^2 (1-\rho^2)}\left(\widetilde{\mu}_i-\log(K))^2+\rho(\widetilde{\mu}_i-\log(K))\right)\right\rbrace}{\sqrt{\sigma_i^2}}\right]}{\mathbb{E}^{\mathbb{Q}}\left[ D(t_i)\displaystyle\frac{e^{-\frac{1}{2}(\widetilde{d}_i(K))^2}}{\sqrt{\sigma_i^2}}\right]}, \end{equation}
where
 \begin{align*}\widetilde{\mu}_i&:=\log S_{t_{i-1}}+\int_{t_{i-1}}^{t_i} r_u \mathrm{d}u -\frac{1}{2} \int_{t_{i-1}}^{t_i}V_u \lambda(S_{t_{i-1}},u)^2 \mathrm{d}u\\
 \widetilde{d}_i(K)&:=\frac{\widetilde{\mu}_i-\log(K)}{\sqrt{\sigma^2_i}},\quad \rho:=\text{corr}(W_{t_i}-W_{t_{i-1}},\log(V_{t_i})|\mathcal{G}_{t_{i-1}}).
\end{align*}
\end{proposition}
\begin{proof}
As in the previous result, by definition we have
$$\mathbb{E}^{\mathbb{Q}}[ D(t_i) V_{t_i}|S_{t_i}=K]=\frac{\mathbb{E}^{\mathbb{Q}}\left[D(t_i) \mathbb{E}^{\mathbb{Q}}[V_{t_i}\delta(S_{t_i}-K)|\mathcal{G}_{t_{i-1}}]\right]}{\mathbb{E}^{\mathbb{Q}}\left[\mathbb{E}^{\mathbb{Q}}[\delta(S_{t_i}-K)|\mathcal{G}_{t_{i-1}}]\right]}.$$
Using \eqref{eq:discreteTime}, we note that in the numerator we have a bivariate Gaussian distribution with the
variables:
$$\begin{pmatrix}\log(S_{t_i})|\mathcal{G}_{t_{i-1}}\\\log (V_{t_i})|\mathcal{G}_{t_{i-1}}\end{pmatrix}\sim\mathcal{N}\left(\begin{pmatrix} \widetilde{\mu}_i \\ \widetilde{\xi}_i \end{pmatrix}, \begin{pmatrix} \sigma_i^2  & \rho\sqrt{\sigma_i^2 \widetilde{\nu}_i}\\  \rho\sqrt{\sigma_i^2 \widetilde{\nu}_i} &  \widetilde{\nu}_i  \end{pmatrix}\right).$$
Computing the conditional expectation in the numerator yields,
$$\mathbb{E}^{\mathbb{Q}}[V_{t_i}\delta(S_{t_i}-K)|\mathcal{G}_{t_{i-1}}]=\frac{\exp\left\lbrace \widetilde{\xi}_i+\frac{1}{2}\widetilde{\nu}_i-\frac{1}{2\sigma_i^2 (1-\rho^2)}\left(\widetilde{\mu}_i-\log(K))^2+\rho(\widetilde{\mu}_i-\log(K))\right)\right\rbrace}{\sqrt{2\pi(1-\rho^2)\sigma_i^2}},$$
and the result readily follows by simplifying terms and proceeding similarly with $\mathbb{E}^{\mathbb{Q}}[ D(t_i)|S_{t_i}=K]$.
\end{proof}
\begin{remark}
Even though we only considered lognormal SV models in Proposition \ref{proposition}, it should be possible to derive analytic expressions with other dynamics as long as the conditional expectations are available in closed-form. This computations, in principle, need be done in a case by case basis.
\end{remark}
\section{The new LSV calibration algorithm}
\begin{enumerate} 
\item Fix a time partition $\{t_0=0,...,t_n=T\}$ and space partition $\{K_0=K_{min},...,K_m=K_{max}\}$.
\item Initialize $\lambda(K_j,t_0)=\displaystyle\frac{\sigma_{Dup}(K_j,t_0)}{V_0},\quad j=0,...,m$ .
\item For $i=1...n:$  \\
\hspace*{0.25cm} For $j=0,...,m$
\begin{quote}
\textbf{Step 1. Predictor:}
\begin{quote}
$\bullet$  Diffuse particles from $t_{i-1}$ to $t_i$ according to \eqref{eq:LSV} and compute $\int_{t_{i-1}}^{t_i} \sqrt{V_u}\lambda(S_{t_{i-1}},u)  dW_u^{||}$\\
$\bullet$ Use formula \eqref{eq:MCexact} to obtain $\displaystyle \frac{\mathbb{E}^{\mathbb{Q}}[D(t_i)V_{t_i}|S_{t_i}=K_j]}{\mathbb{E}^{\mathbb{Q}}[D(t_i)|S_{t_i}=K_j]}$
\end{quote}
\textbf{Step 2. Corrector:}
\begin{quote}
$\bullet$$\lambda^2(K_j,t_i)=\displaystyle\left(\sigma_{Dup}(K_j,t_i)^2-\frac{\mathbb{E}^{\mathbb{Q}}\left[ D(t_i) \left(r_{t_i}-\bar{r}_{t_i}\right)\displaystyle \ind_{\{S_{t_i}>K_j\}}\right]}{\frac{1}{2}K\frac{\partial^2 C}{d K^2}}\right)\frac{\mathbb{E}^{\mathbb{Q}}[ D(t_i)|S_{t_i}=K_j]}{\mathbb{E}^{\mathbb{Q}}[ D(t_i)V_{t_i}|S_{t_i}=K_j]}.$\\
$\bullet$ Extrapolate flat outside $[K_{min},K_{max}]$ and lay a cubic spline in between.
\end{quote}
\end{quote}
\end{enumerate}
\begin{remark}
If $V$ is lognormal one can replace \eqref{eq:MCexact} by \eqref{eq:CloseFormLognormal}.
\end{remark}\bigskip
\textbf{Computational considerations:}\\
While one can consider sorting particles according to spot value as in Guyon and Henry-Labord\`ere \cite{Guyon-Labordere}, ultimately the savings in computational time will depend on the Monte Carlo sample size and the relative cost/benefit of the sorting procedure. In our case, it is preferable to sort particles according to $\frac{\mu_i}{\sqrt{\sigma_i^2}}$ to set a threshold. Whether sorting particles or not, the computational complexity of our algorithm will be virtually the same as the original particle method.\\\\
As mentioned in \cite{Guyon-Labordere} with the sorting technique in place, the cost of computing  the term
$$\mathbb{E}^{\mathbb{Q}}\left[ D(t_i)\left(r_{t_i}-\bar{r}_t\right)\displaystyle \ind_{\{S_{t_i}>K_j\}}\right]$$ may overtake the computational complexity. Indeed for low strikes $\ind_{\{S_{t_i}>K_j\}}$ will be one with high likelihood and particles cannot be disregarded. Contrary to the Malliavin representation approach presented in \cite{Guyon-Labordere}, be propose the use of symmetry and Theorem \ref{thm:condDistr} to obtain (see Appendix \ref{App:Proof1})
\begin{align}\mathbb{E}^{\mathbb{Q}}\left[ D(t_i)\left(r_{t_i}-\bar{r}_t\right)\displaystyle \ind_{\{S_{t_i}>K_j\}}\right]&=\mathbb{E}^{\mathbb{Q}}\left[ D(t_i)\left(r_{t_i}-\bar{r}_t\right)\Phi(-d_i(K_j))\right]\label{eq:1}\\&=-\mathbb{E}^{\mathbb{Q}}\left[ D(t_i)\left(r_{t_i}-\bar{r}_t\right)\Phi(d_i(K_j))\right]\label{eq:2},\end{align}
where $\Phi(\cdot)$ is the standard Gaussian cdf. Therefore, we can use expression \eqref{eq:1} for $K>S_0$ and \eqref{eq:2} for $K\leq S_0$ to disregard particles and keep the computational cost controlled. We further emphasise that when extrapolating, if $K\to 0$ or $K\to\infty$ the expectations above converge to zero. \\\\
\textbf{Theoretical considerations:}\\Precise conditions on system \eqref{eq:LSV}-\eqref{eq:LeverageFunctionCondition}  for a solution to exists still remains an open question. In \cite{CMR,Guyon-Labordere} it is reported that for large values of vol of vol the algorithm fails to converge.  We tested this phenomena, but seems that the domain of convergence of our algorithm is superior to that of the particle method see Figure \ref{Fig:conv}. Whether such behaviour is caused by the choice of an inadequate bandwidth parameter $h$ remains unclear, though raises the issue of instability in the original method.  Further tests with higher level of vol showed that both our algorithm and the original particle method fail (see Figure \ref{Fig:fail}) to match the market smiles. The precise theoretical identification of an upper bound in volatility of volatility remains unanswered.\\\\
\begin{figure}
\centering
\includegraphics[scale=0.35]{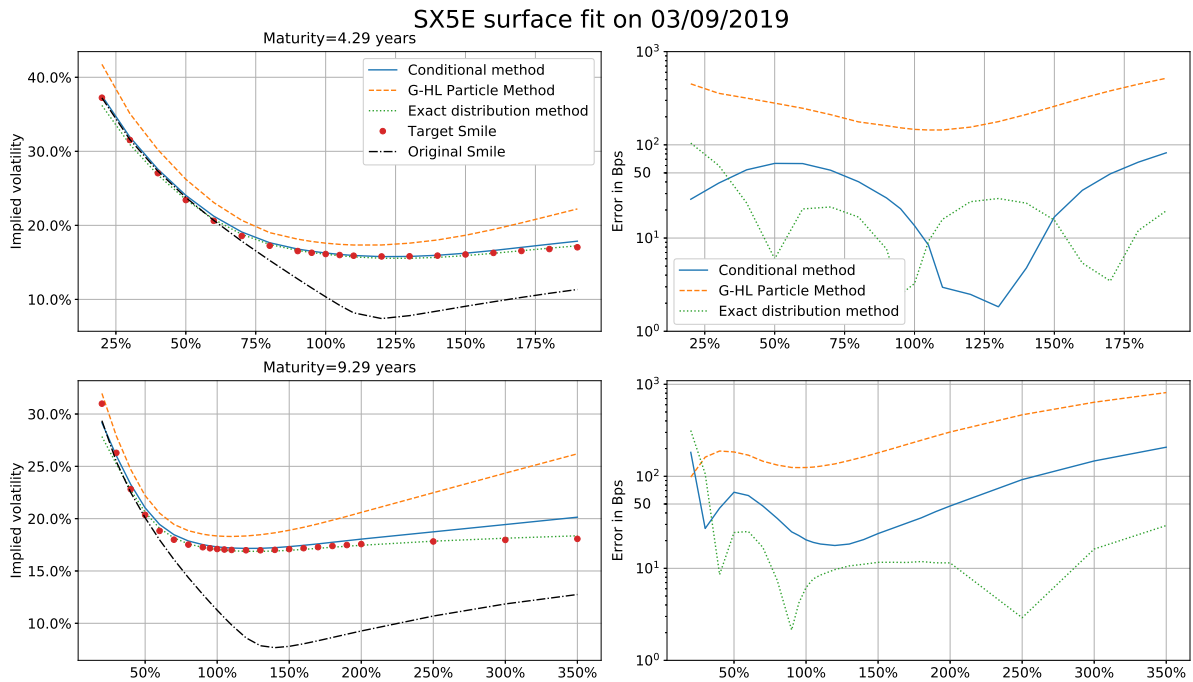}
\caption{SX5E (3-Sept-2019) Implied volatility surface. Rough volatility parameters: $H=0.2$, $\rho_{WZ}=-80\%,\;\beta=0.5,\;\nu=380\%$. Vasicek parameters: $\kappa=1,\;\sigma=0.5\%,\;\rho_{WY}=0\%,\;r_0=1.5\%,\;\rho_{ZY}=0\%$. The conditional and exact distribution methods are given by equations \eqref{eq:MCexact} and \eqref{eq:CloseFormLognormal} respectively. G-HL Particle method is described in \eqref{eq:ParticleMethod}.}
\label{Fig:conv}
\end{figure}
\begin{figure}
\centering
\includegraphics[scale=0.35]{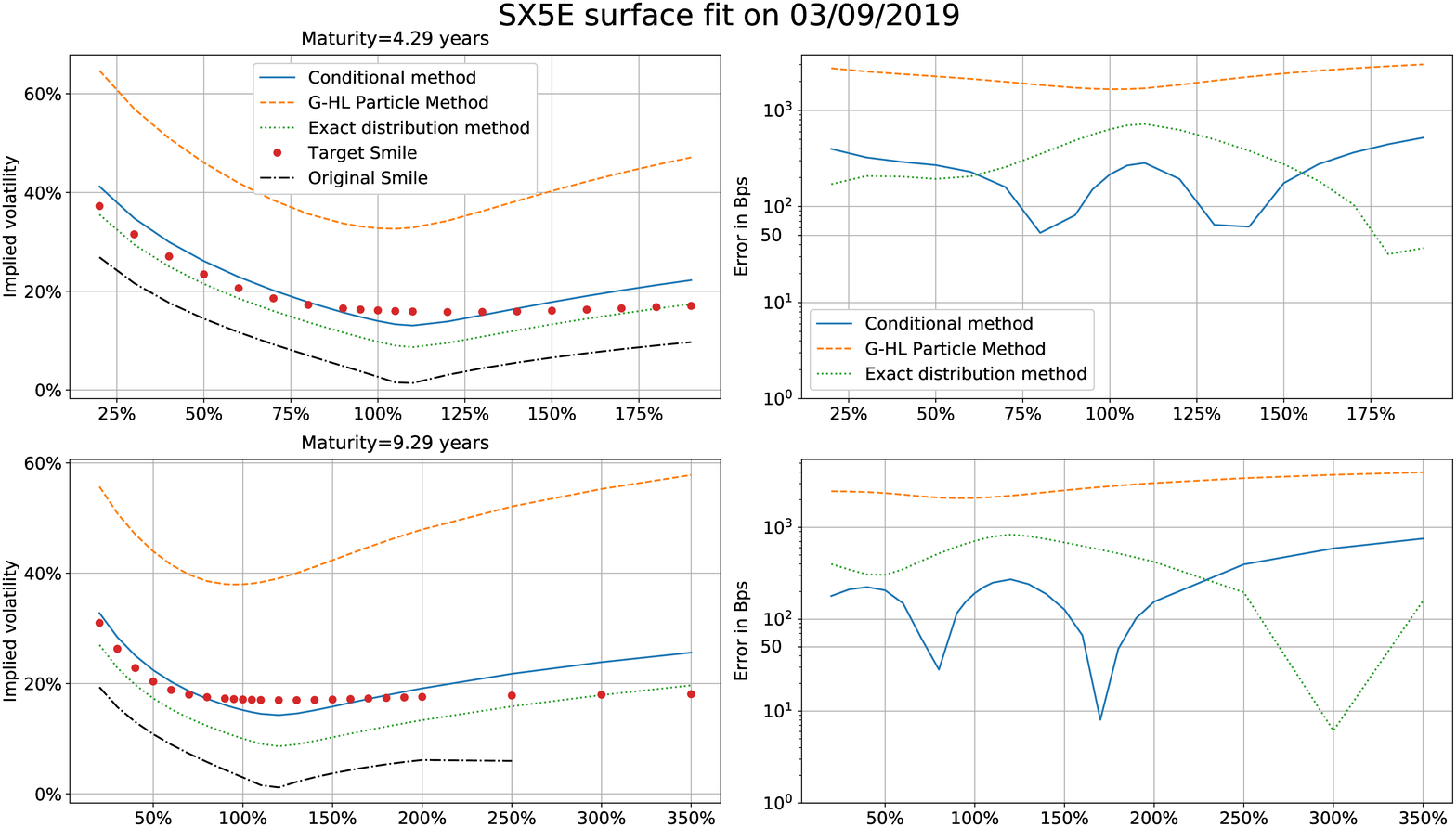}
\caption{SX5E (3-Sept-2019) Implied volatility surface. Rough volatility parameters: $H=0.2$, $\rho_{WZ}=-80\%,\;\beta=0.5,\;\nu=800\%$. Vasicek parameters: $\kappa=1,\;\sigma=0.5\%,\;\rho_{WY}=0\%,\;r_0=1.5\%,\;\rho_{ZY}=0\%$. The conditional and exact distribution methods are given by equations \eqref{eq:MCexact} and \eqref{eq:CloseFormLognormal} respectively. G-HL Particle method is described in \eqref{eq:ParticleMethod}.}
\label{Fig:fail}
\end{figure}
\section{Overture to local correlation and basket smiles}
So far, we have only considered the mono-underlying case, but there is no reason why our method cannot be extended to a multi-asset environment. Let us consider an index $I_t=\sum_{i=1}^p \omega_i S^i_t$, where
\begin{align*}
\frac{\mathrm{d}I_t}{I_t}&=r_t  \mathrm{d}t+\sum_{i=1}^n \frac{\omega_i S^i_t  \lambda^i(S^i_t,t) \sqrt{V^i_t}}{I_t} \mathrm{d} W^i_t\label{eq:LSV}\\
\frac{\mathrm{d}S^i_t}{S^i_t}&=r_t \mathrm{d}t+\lambda^i(S^i_t,t) \sqrt{V^i_t} \mathrm{d}W^i_t,\quad i=1,...,p\\
\mathrm{d}[W^i,W^j]_t:&= \rho_{ij}(I_t,t) \mathrm{d}t .
\end{align*}
The condition for the index (or basket) to be calibrated to market smiles is given by (see Guyon \cite{Guyon}):
\begin{equation}\sum_{i=1}^n \sum_{j=1}^n \omega_i\omega_j  \rho_{ij}(K,t) \frac{\mathbb{E}^{\mathbb{Q}}[D(t) \lambda^i(S^i_t,t) \sqrt{V^i_t}S_t^i \lambda^i(S^j_t,t) \sqrt{V^j_t} S_t^j|I_t=K]}{K^2\mathbb{E}^{\mathbb{Q}}[D(t)|I_t=K]}=\sigma_{Dup}(K,t)^2-\frac{\mathbb{E}^{\mathbb{Q}}\left[D(t)\left(r_t-\bar{r}_t\right)\ind_{\{I_t>K\}}\right]}{\frac{1}{2}K\frac{\partial^2 C}{d K^2}}.\label{eq:multiAsset}\end{equation}
At this stage an assumption on the particular structure of $\rho_{ij}(\cdot,\cdot)$ is tipycally made, both to ensure computational feasibility and positive definitness of the correlation matrix. Nevertheless, for the purpose of numerically solving \eqref{eq:multiAsset}, the problem is equivalent to being able to compute 
$$\frac{\mathbb{E}^{\mathbb{Q}}[D(t) \lambda^i(S^i_t,t) \sqrt{V^i_t}S_t^i \lambda^i(S^j_t,t) \sqrt{V^j_t} S_t^j|I_t=K]}{\mathbb{E}^{\mathbb{Q}}[D(t)|I_t=K]},\quad i,j=1,..,n.$$
Both Corollary \ref{corollary} and Proposition \ref{proposition} can be easily adapted to this setting; one needs to find the appropriate filtration $\mathcal{G}^I_t$ to make the tuple $(S^i_{t_i},S^j_{t_i})|\mathcal{G}^I_{t_{i-1}}$ jointly lognormal.  Therefore, similar closed-form and exact expressions  can be obtained after some tedious (but not difficult) calculations.
\section{Numerical tests with hybrid models}\label{sec:numerics}
\subsection{Vasicek and rough local stochastic volatility}
We first consider a rough volatility model, similar to the one introduced by Bayer, Friz and Gatheral \cite{BFG} topped with the Vasicek stochastic interest rate model. It is well known (see Al\`os, Le\'on and Vives \cite{ALV}) that such models reproduce a power law decay on the short time ATM skew. However, to control the long term behaviour we also add a mean-reversion parameter $\beta$. It is worth noting that non-Markovian dynamics imply that past behaviour of the volatility process influences the future behaviour.
\begin{align*}\frac{\mathrm{d}S_t}{S_t}&=r_t \mathrm{d}t+\lambda(S_t,t) \sqrt{V_t} \mathrm{d}W_t\\
V_t&=\xi_0(t)\mathcal{E}\left(\nu\sqrt{2H}\int_0^t (t-s)^{H-1/2}e^{-\beta(t-s)}\mathrm{d}Z_s\right),\quad \beta,\nu>0,\; H\in(0,1/2)\\
\mathrm{d}r_t&=(r_0-\kappa r_t) \mathrm{d}t +\sigma \mathrm{d}Y_t,\\
\mathrm{d}[W,Z]_t&=\rho_{WZ}\mathrm{d}t,\quad \mathrm{d}[W,Y]_t=\rho_{WY}\mathrm{d}t,\quad \mathrm{d}[Z,Y]_t=\rho_{ZY}\mathrm{d}t.
\end{align*}
where $\mathcal{E}(x)=\exp(x-\frac{1}{2}Var(x))$. Due to its non-Markovian nature, there is no forward Kolmogorov equation known for the transition
density. Hence PDE methods are out of the picture and one can only resort to simulation based
methods (see Horvath, Jacquier and Muguruza \cite{HJM} for details on simulation). In Figure \ref{fig:roughFit} we report the LSV calibration results with $50.000$ Monte Carlo paths and
$\frac{1}{252}$ time step. We observe that the proposed algorithm (both with the conditional approach and
distribution approach) converges appropriately and performs at the level of the particle method \cite{Guyon-Labordere}. For all tested maturities, we obtain an accuracy of a few basis points for reasonable strikes. We do not observe a significant performance difference between \eqref{eq:MCexact} and \eqref{eq:CloseFormLognormal} and conclude that both implementations yield a similar result.
\begin{figure}
\centering
\includegraphics[scale=0.35]{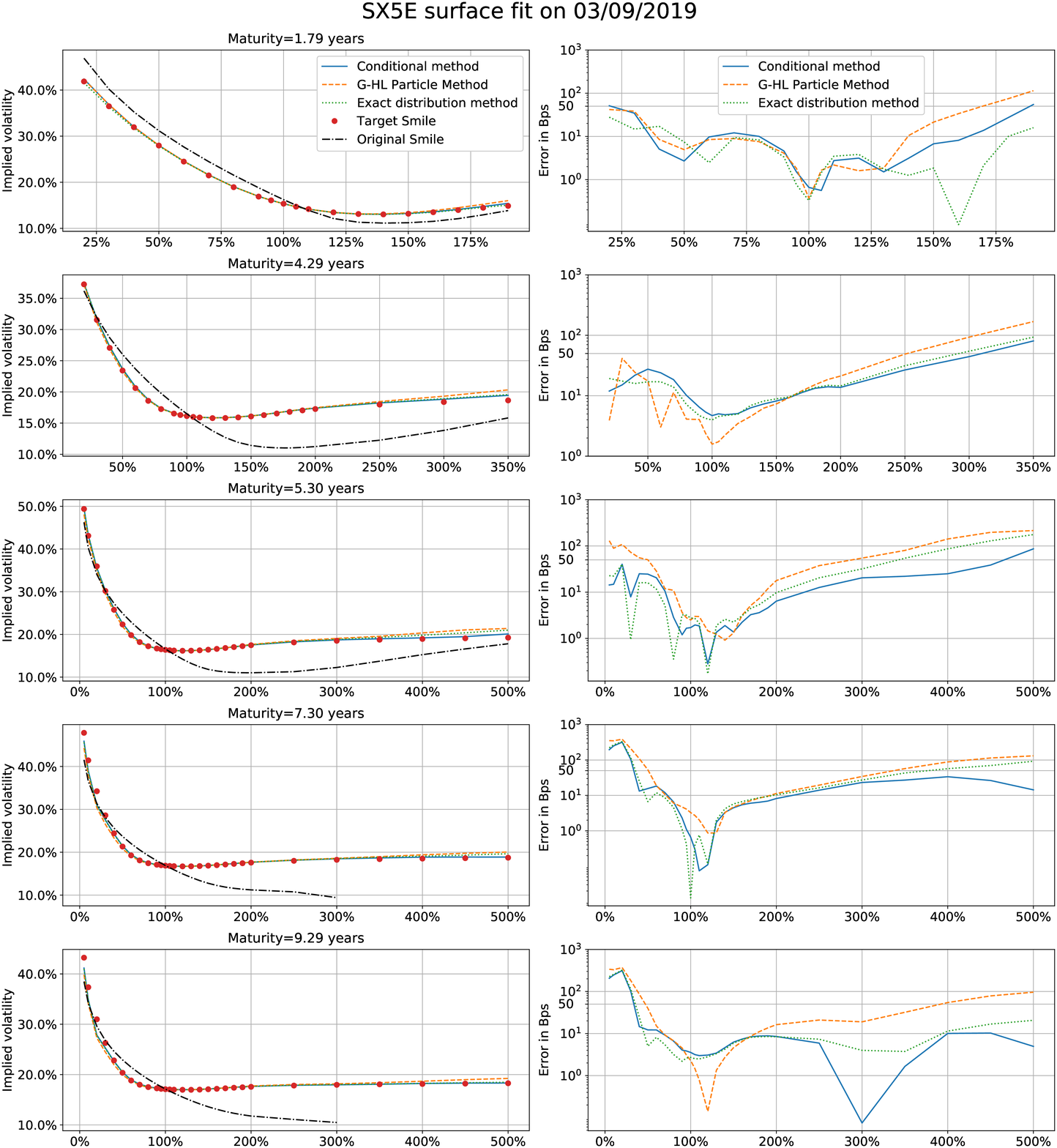}
\caption{SX5E (3-Sept-2019) Implied volatility surface. Rough volatility parameters: $H=0.2$, $\rho_{WZ}=-80\%,\;\beta=0.5,\;\nu=200\%$. Vasicek parameters: $\kappa=1,\;\sigma=0.5\%,\;\rho_{WY}=0\%,\;r_0=1.5\%,\;\rho_{ZY}=0\%$. The conditional and exact distribution methods are given by equations \eqref{eq:MCexact} and \eqref{eq:CloseFormLognormal} respectively. G-HL Particle method is described in \eqref{eq:ParticleMethod}.}
\label{fig:roughFit}
\end{figure}
\subsection{Vasicek and local 2F Bergomi}
To test our method in a higher dimensional  setting, we consider the two Factor Bergomi \cite{Bergomi} model with stochastic interest rates given by a Vasicek
interest rate model. While Bergomi's model is a popular equity model due to its flexibility to fit dynamical properties of the smile,  it is driven by a 4 dimensional Brownian, which poses a remarkable challenge on a PDE framework.
\begin{align*}\frac{\mathrm{d}S_t}{S_t}&=r_t \mathrm{d}t+\lambda(S_t,t) \sqrt{V_t} \mathrm{d}W_t\\
V_t&=\xi_0(t)\mathcal{E}\left(\nu\alpha_\theta\left( (1-\theta) \int_0^t e^{-\kappa_X(t-s)}\mathrm{d}X_s+\theta \int_0^t e^{-\kappa_Y(t-s)}\mathrm{d}Y_s \right)\right),\quad \kappa_X,\kappa_Y,\nu>0,\; \theta\in[0,1]\\
\mathrm{d}r_t&=(r_0-\kappa r_t) \mathrm{d}t +\sigma\mathrm{d}Z_t
\end{align*}
 In Figure \ref{fig:2FFit} we report the results with $50.000$ Monte Carlo
paths and $\frac{1}{252}$ time step. The results show again that our algorithm performs as good as the particle
method and converges successfully, with accuracies of few basis points for relevant strikes as before.

\begin{figure}
\centering
\includegraphics[scale=0.35]{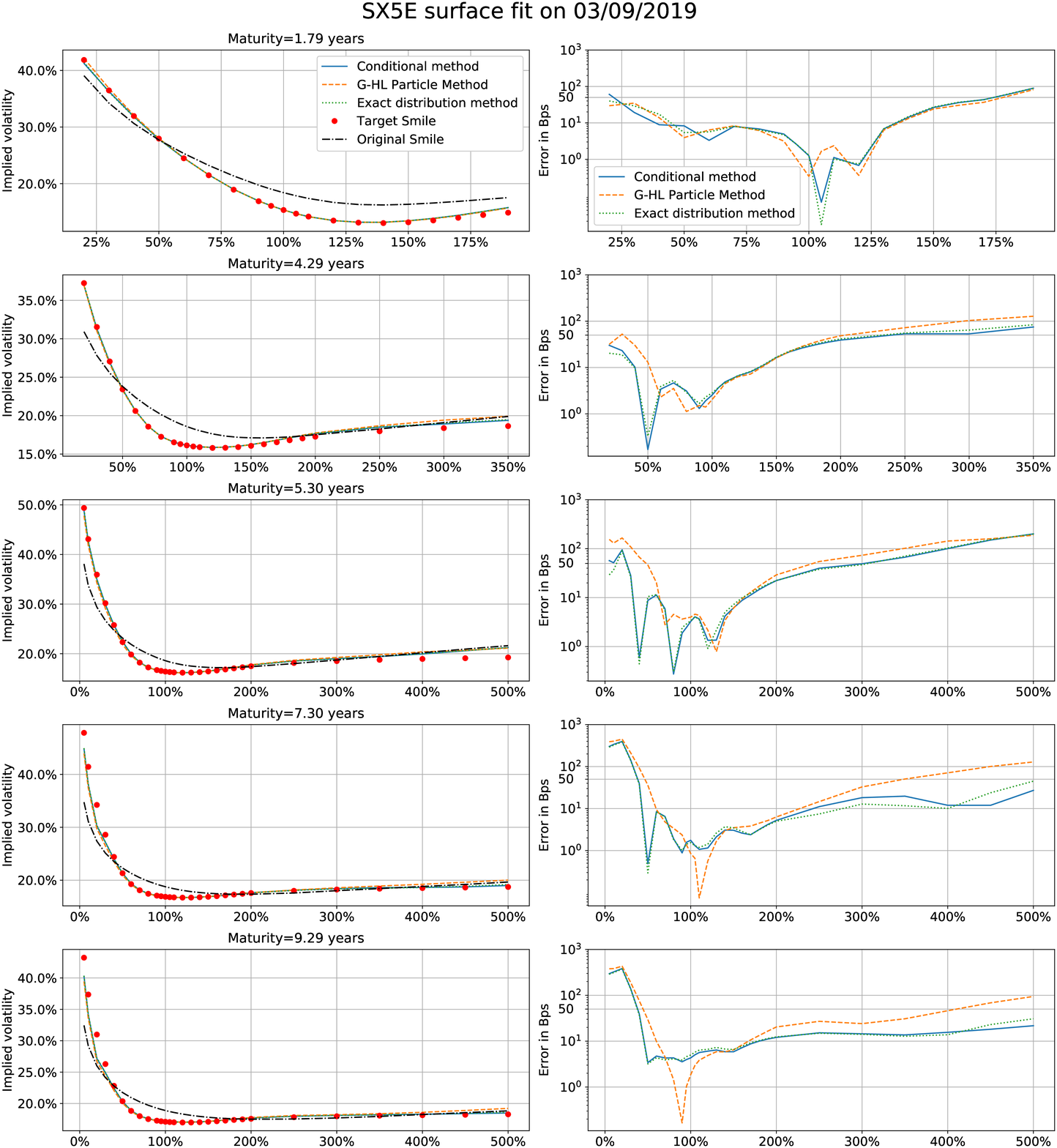}
\caption{SX5E (3-Sept-2019) Implied volatility surface. 2F Bergomi parameters: $\kappa_X=0.5,\;\kappa_Y=8,\;\theta=80\%$, $\rho_{WX}=-10\%,\;\rho_{WY}=-80\%,\;\rho_{XY}=30\%,\;\nu=400\%,\;\rho_{WZ}=\rho_{XZ}=\rho_{YZ}=0\%$. Vasicek parameters: $\kappa=1,\;\sigma=0.5\%,\;r_0=1.5\%$. The conditional and exact distribution methods are given by equations \eqref{eq:MCexact} and \eqref{eq:CloseFormLognormal} respectively. G-HL Particle method is described in \eqref{eq:ParticleMethod}.}
\label{fig:2FFit}
\end{figure}
\section{Summary}
In this article we have presented a new Monte Carlo method to calibrate the leverage function to any stochastic volatility model. Building upon the particle method developed by Guyon and Henry-Labord\`ere \cite{Guyon-Labordere}, we obtained a closed-form and exact calibration procedure that does not depend on any external parameter fine-tuning nor incurs any additional computational cost. Our experiments show that the accuracy of the method is in line with that of the original particle method and improves the convergence boundary in vol of vol. Overall, this new approach allows for an easier, safer and more robust implementation in an automatised production level environment.

\appendix
\section{Proof 1}\label{App:Proof1}
We have
\begin{align*}
&\mathbb{E}^{\mathbb{Q}}\left[ D(t_i)\left(r_{t_i}-\bar{r}_t\right) \ind_{\{S_{t_i}>K_j\}}\right]=\mathbb{E}^{\mathbb{Q}}\left[ D(t_i)\left(r_{t_i}-\bar{r}_t\right)\mathbb{E}^{\mathbb{Q}}\left[ \ind_{\{S_{t_i}>K_j\}}|\mathcal{F}^{W}_{t_{i-1}}\cup \mathcal{F}^{Z}_{t_i}\right]\right].
\end{align*}
Now, we use Theorem \ref{thm:condDistr} to obtain that $S_i|\mathcal{F}^{W}_{t_{i-1}}\cup \mathcal{F}^{Z}_{t_i}$ is lognormal and obtain
$$\mathbb{E}^{\mathbb{Q}}\left[ D(t_i)\left(r_{t_i}-\bar{r}_t\right)\displaystyle \ind_{\{S_{t_i}>K_j\}}\right]=\mathbb{E}^{\mathbb{Q}}\left[ D(t_i)\left(r_{t_i}-\bar{r}_t\right)\Phi(-d_i(K_j))\right].$$
Additionally, we note

$$0=\mathbb{E}^{\mathbb{Q}}\left[ D(t_i)\left(r_{t_i}-\bar{r}_t\right)\right]=\mathbb{E}^{\mathbb{Q}}\left[ D(t_i)\left(r_{t_i}-\bar{r}_t\right) \ind_{\{S_{t_i}>K_j\}}\right]+\mathbb{E}^{\mathbb{Q}}\left[ D(t_i)\left(r_{t_i}-\bar{r}_t\right) \ind_{\{S_{t_i}\leq K_j\}}\right].$$
Thus,
$$\mathbb{E}^{\mathbb{Q}}\left[ D(t_i)\left(r_{t_i}-\bar{r}_t\right) \ind_{\{S_{t_i}>K_j\}}\right]=-\mathbb{E}^{\mathbb{Q}}\left[ D(t_i)\left(r_{t_i}-\bar{r}_t\right) \ind_{\{S_{t_i}\leq K_j\}}\right]=-\mathbb{E}^{\mathbb{Q}}\left[ D(t_i)\left(r_{t_i}-\bar{r}_t\right)\Phi(d_i(K_j))\right].$$

\end{document}